\documentclass[runningheads]{llncs}
\pdfoutput=1
\usepackage{times}
\usepackage{xspace}
\usepackage{url}
\usepackage{color}
\usepackage{amssymb}
\usepackage{multirow}
\usepackage[linesnumbered,vlined]{algorithm2e}
\usepackage{enumerate}
\usepackage{paralist}
\usepackage{microtype}
\usepackage{wrapfig}

\title{Automated Benchmarking of Incremental \\ SAT and QBF
  Solvers\thanks{This work was supported by the Austrian Science Fund (FWF)
    under grant S11409-N23. This article will appear in the
      \textbf{proceedings} of the \emph{20th International Conference on
Logic for Programming, Artificial Intelligence and Reasoning (LPAR)}, LNCS,
Springer, 2015.}}

\author{Uwe Egly \and Florian Lonsing \and Johannes Oetsch}

\institute{
Vienna University of Technology,\\
Institute of Information Systems 184/3,\\
Favoritenstra\ss{}e\ 9-11, A-1040 Vienna, Austria \\
\email{\{uwe,lonsing,oetsch\}@kr.tuwien.ac.at}
}

\newcommand{\depqbf}{\ensuremath{\mathtt{DepQBF}}\xspace}

\newcommand{\qube}{\ensuremath{\mathtt{QuBE}}\xspace}
\newcommand{\minisat}{\ensuremath{\mathtt{MiniSAT}}\xspace}
\newcommand{\lingeling}{\ensuremath{\mathtt{Lingeling}}\xspace}
\newcommand{\picosat}{\ensuremath{\mathtt{PicoSAT}}\xspace}
\newcommand{\aigbmc}{\ensuremath{\mathtt{aigbmc}}\xspace}
\newcommand{\ipasir}{IPASIR\xspace}

\newcommand{\quant}{\mathit{quant}}
\newcommand{\push}{\mathtt{push}}
\newcommand{\pop}{\mathtt{pop}}
\newcommand{\add}{\mathtt{add}}

\begin{document}

\maketitle

\begin{abstract}
Incremental SAT and QBF solving potentially yields improvements when sequences
of related formulas are solved. An incremental application is usually tailored towards some specific
solver and decomposes a problem into incremental solver calls.
 This hinders the independent
comparison of different solvers, particularly when the application
program is not available. As a remedy, we present an approach to
automated benchmarking of incremental SAT and QBF solvers. Given a collection
of formulas in (Q)DIMACS format generated incrementally by an application
program, our approach automatically translates the formulas into instructions
to import and solve a formula by an incremental SAT/QBF solver.  The result of
the translation is a program which replays the incremental solver calls and thus
allows to evaluate incremental solvers independently
from the application program. We
illustrate our approach by different hardware verification problems for SAT
and QBF solvers.
\end{abstract}

\section{Introduction}{\label{sec:intro}}

Incremental solving has contributed to the success of SAT technology and
potentially yields considerable improvements in applications where sequences
of related formulas are solved.
The logic of quantified Boolean formulas (QBF) extends propositional logic
(SAT) by explicit existential and universal quantification of
variables and lends itself for problems within PSPACE. 
Also for QBFs, incremental solving has been successfully applied in different domains~\cite{DBLP:conf/fmcad/BloemEKKL14,DBLP:conf/aisc/EglyKLP14,DBLP:conf/date/MarinMLB12,DBLP:journals/aicom/MillerMB15}.

The development of SAT and QBF solvers has been driven by competitive events
like the SAT Competitions,  
QBF Evaluations (QBFEVAL), or the QBF Galleries. 
 These events regularly result in publicly available benchmarks submitted by the
participants which help to push the state of the art in SAT and QBF solving.
In the past,  
the focus was on \emph{non-incremental} SAT solving, and
 the evaluation of \emph{incremental} solvers does not
readily benefit from competitions and available benchmark collections. 

Benchmarking incremental solvers requires to solve a sequence of related
formulas. To this end, the formulas must be incrementally imported to the
solver and solved by means of API calls. The API calls are typically generated
by an application program, like a model
checker or a 
 formal verification 
or planning tool, for example, 
which tackles
a problem by encoding it incrementally to a sequence of formulas. In order to
compare different incremental solvers on that sequence of formulas, the
solvers must be tightly coupled with the application program by linking them
as a library. 
Hence benchmarking of incremental solvers relies on the application program
used to generate the sequence of formulas which, however, often is not 
available. Even if the application program is available, it has to
be adapted to support different solvers, where each solver might come
with its own API. Further, the same sequence of formulas must be generated
multiple times by the application program to compare different solvers.

To remedy this situation, we present an approach to automated benchmarking of
incremental SAT and QBF solvers which decouples incremental SAT/QBF
solving from incremental generation of formulas using an application
program. This is achieved by translating a sequence of related
CNFs and QBFs in prenex CNF (PCNF) into API calls of incremental solvers. 
Such a sequence might be the output of an application program or it was taken from existing
benchmark collections.
The formulas 
are then syntactically analyzed and  
instructions to
incrementally import and solve them  
are generated.
For CNFs, the instructions are  function calls in the IPASIR
API, which has been proposed for the Incremental Library Track of the SAT Race
2015.\footnote{\url{http://baldur.iti.kit.edu/sat-race-2015/}} For PCNFs, the
instructions correspond to calls of the API of the QBF solver
\depqbf,\footnote{\url{http://lonsing.github.io/depqbf/}} which generalizes IPASIR and allows
to update quantifier prefixes. 
The result of translating a sequence 
of formulas to solver API calls is a \emph{standalone benchmarking program}
which replays the incremental solver calls. 
Any incremental SAT/QBF solver supporting the
IPASIR API or its QBF extension as implemented in \depqbf can be 
integrated by simply linking it to the program. This allows to 
compare different solvers independently from an application.

In some applications, the sequence of formulas depends on the used solver, e.g.,
if truth assignments are used to guide the process. 
Even then, our approach allows to
compare different incremental solvers on the fixed sequences generated with one particular
solver. However, then it is important to note that this comparison is limited to this particular fixed sequence, it
would be unfair to conclude something about the performance of the solvers would they have been
genuinely used within the application.
This problem occurs also in sequences of formulas which are already present in benchmark collections.
For experiments in this paper,  we only considered applications where the sequences of generated formulas do not depend on 
intermediate truth assignments.
 
As 
our approach is also applicable to already generated formulas
that are part of existing benchmark collections,
such collections become available to developers of incremental solvers. 
Furthermore, comparisons between solvers in incremental and non-incremental mode are
made possible. 
In addition, since the input for the benchmarking program describes only
the differences between consecutive formulas, we obtain a quite succinct representation of incremental benchmarks.
Our approach to automated benchmarking of incremental SAT and QBF solvers
underpins the goal of the Incremental Library Track of the SAT Race 2015.
We have generated benchmarks and submitted them to this competition.

\section{Background}{\label{sec:bg}}

We consider propositional formulas in CNF and 
identify a CNF with the set of its clauses. 
A sequence $\sigma= (F_{1}, 
\ldots, F_{n})$ of formulas represents the formulas that are
incrementally generated and solved by an application program. 
A QBF $\psi = P.F$ in prenex CNF (PCNF)  extends a CNF $F$
by a quantifier prefix $P$. 
The prefix $P =
Q_{1}, \ldots, Q_{n}$ of a QBF is a
sequence of pairwise disjoint \emph{quantified sets} $Q_i$.  A {quantified set} $Q$ is a set of
variables with an associated quantifier $\quant(Q) \in \{\exists,
\forall\}$. 
We consider only closed PCNFs. 
For adjacent quantified sets $Q_i$ and $Q_{i+1}$, 
$\quant(Q_i) \not = \quant(Q_{i+1})$. 
Given a prefix $P = Q_{1},
\ldots, Q_{n}$, index $i$ 
is the \emph{nesting level} of $Q_{i}$ in $P$. 

Our automated benchmarking approach   is 
based on \emph{solving under assumptions}~\cite{DBLP:conf/sat/EenS03,DBLP:journals/entcs/EenS03} 
as implemented in modern
SAT~\cite{DBLP:conf/sat/AudemardLS13,DBLP:conf/sat/LagniezB13,DBLP:conf/sat/NadelRS14}
and QBF
solvers~\cite{DBLP:conf/cp/LonsingE14,DBLP:conf/date/MarinMLB12,DBLP:journals/aicom/MillerMB15}.
When
solving a CNF under assumptions, the clauses are augmented with \emph{selector
variables}. Selector variables allow for
temporary variable assignments made by the user via the solver API. If the
value assigned to a selector variable satisfies the clauses where it
occurs, then these clauses are effectively removed from the CNF.
This way, the user controls which clauses
appear in the CNF in the forthcoming incremental solver run.
The IPASIR API proposed for the Incremental Library Track of the SAT Race 2015
consists of a set of functions for adding clauses to a CNF and
handling assumptions. 
A disadvantage of this approach is that the user has to keep track of the used selector variables 
and assumptions manually.

For incremental QBF solving, additional API functions are needed to remove
quantified sets and variables from and add them to a prefix. For 
QBF solvers, we generate
calls in the API of \depqbf
which generalizes IPASIR by functions to manipulate quantifier
prefixes. Additionally, it allows to remove and add clauses in a
stack-based way by push/pop operations
where selector variables and assumptions are handled internal to the solver and hence are invisible to the user~\cite{DBLP:conf/cp/LonsingE14}. 
 For details on the IPASIR and \depqbf interfaces, we refer to the respective webpages mentioned in the introduction.

\section{Translating Related Formulas into Incremental Solver Calls}{\label{sec:workflow}}

We present the workflow to translate a given sequence
$\sigma = (\psi_{1}, \ldots, \psi_{n})$ of related (P)CNFs into a
standalone benchmarking program which calls an integrated solver via its API
to incrementally solve the formulas from $\psi_{1}$ up to $\psi_{n}$:

\begin{compactenum}
\item{
   First, the formulas in $\sigma$ are 
  analyzed and the syntactic differences between each $\psi_{i}$ and $\psi_{i+1}$ are
  identified. 
  This includes clauses and quantified sets that have to be added or removed
  to obtain $\psi_{i+1}$ from $\psi_{i}$. Also, variables may be added to or removed from quantified sets.
For CNFs, 
the prefix
  analysis is omitted.}

\item{
The
  differences between the formulas identified  in the first step are expressed by generic update instructions  
  and are written to a file. 
   A clause set is represented as a stack which can be updated via push and pop operations. 
The update instructions for quantifier prefixes are adding a quantified set at a
nesting level and adding new variables to quantified sets
already present in the prefix.  Unused variables are deleted from the prefix be the solver.
}

\item{
    Files that contain generic update
  instructions are then interpreted by a \emph{benchmarking program} which
translates them into calls of the IPASIR
  API (for CNFs) or QBF solver calls (for PCNFs). For the
  latter, calls of \depqbf's API are generated. 
}
\end{compactenum}

\noindent

\noindent
The {benchmarking program} is standalone and independent from the
application program used to generate $\sigma$. It takes the
files containing the generic update instructions as the only input. 
Multiple solvers may be integrated in the benchmarking
program by linking them as libraries.  
Files containing the update instructions can serve as standardised benchmarks for
incremental SAT and QBF solvers.

\paragraph{Analyzing CNFs.} \label{sec:analyze:cnfs}

The algorithm to analyze sequences $\sigma = (F_{1}, \ldots,
F_{n})$ of clause sets relies on a stack-based representation of
$F_i$
which allows for simple deletion of
clauses that have been added most recently. 
A clause $c$ which appears in some $F_i$  and
is removed later at some point to obtain $F_j$ with $i < j \leq n$ is called
\emph{volatile in} $F_{i}$. A clause which appears in some $F_i$ for
the first time and also appears in every $F_j$ with $i < j \leq n$ and
hence is never deleted is called \emph{cumulative in} $F_{i}$.

The algorithm to analyze sequence $\sigma$ identifies volatile and cumulative clauses in all clause
sets in $\sigma$. 
Cumulative clauses
are pushed first on the stack representing the current clause set
because they are not removed anymore after they have been added. 
Volatile clauses are pushed last because they are
removed at some point by a pop operation when constructing a later
formula in $\sigma$. 
For illustration, consider the following sequence $\sigma = (F_1, \ldots, F_4)$ of clause sets $F_{i} $ along with their respective sets $C_{i}$ of cumulative clauses and sets $V_{i}$ of volatile clauses: 
\begin{center}
\small
\begin{tabular}{l@{\hspace{1cm}}l@{\hspace{1cm}}l}
$F_{1} = \{ c_{1}, c_{2}, v_{1}  \}	$			& $C_{1} = \{  c_{1}, c_{2} \}$	& $V_{1} = \{ v_{1}  \}		$	\\
$F_{2} = \{ c_{1}, c_{2}, c_{3}, v_{1}, v_{2}  \}	$		& $C_{2} = \{c_{3}\}$		& $V_{2} = \{  v_{1}, v_{2} \}$		\\
$F_{3} = \{ c_{1}, c_{2}, c_{3}, c_{4}, v_{1}, v_{3}  \}$		& $C_{3} = \{  c_{4} \}$		& $V_{3} = \{   v_{1}, v_{3} \}$		\\
$F_{4} = \{ c_{1}, c_{2}, c_{3}, c_{4}, c_{5}  \}		$	& $C_{4} = \{  c_{5}\}	$		& $V_{4} = \emptyset	$		\\
\end{tabular}
\end{center}

After the sets of cumulative and volatile clauses have been identified
for each $F_i$, the clause sets 
can be incrementally constructed by means of the following operations on
the clause stack:
adding a set $C$ of clauses permanently to a formula by $\add(C)$,
pushing a set $C$ of clauses on the stack by $\push(C)$, and
popping a set  of clauses from the stack by $\pop()$.
The sequence $\sigma = (F_{1}, \ldots, F_{4})$ from the example above is generated incrementally by executing the following stack operations:
\begin{center}
\small
\begin{tabular}{l@{\hspace{1cm}}l@{\hspace{1cm}}l}
& $\add(C_{1})$ & $\push(V_{1})$ \\
$\pop()$ & $\add(C_{2})$ & $\push(V_{2})$ \\
$\pop()$ & $\add(C_{3})$ & $\push(V_{3})$ \\
$\pop()$ & $\add(C_{4})$ & $\push(V_{4})$ \\
\end{tabular}
\end{center}

Note that the above schema of stack operations  generalises to
\emph{arbitrary} sequences of clause
sets,  i.e., we need 
at most one push, one add,  and one pop operation in each step,
provided that the clauses have been classified as volatile or
cumulative  before.  

The algorithm for identifying cumulative and
volatile clauses in a sequence of clause sets appears as  
Algorithm~\ref{alg:cumvol}.
For SAT solvers supporting the
IPASIR API, stack frames for volatile clauses pushed on the clause stack are implemented by
selector variables.
Our current implementation of the benchmarking program includes
\depqbf as the only incremental QBF solver  which supports push/pop
operations natively via its API~\cite{DBLP:conf/cp/LonsingE14}.
Note that the relevant part of the input that potentially limits scalability of Algorithm~\ref{alg:cumvol} is the number of variables and clauses in the formulas. The number of formulas  is usually relatively low. The operations on clause sets are implemented such that set intersection and difference are in $O(m \cdot \mathit{log}\,m)$, searching an element is in $O(m)$, and adding or deleting elements are in $O(1)$, where $m$ is the maximal number of clauses in any formula.

\begin{algorithm}[t]
\small
\SetKwInOut{Input}{Input}\SetKwInOut{Output}{Output}
\SetKw{Move}{move}
\SetKw{from}{from}
\SetKw{Break}{break}

\Input{Clause sets $F_1$, $F_2$, \ldots, $F_n$   (at least two sets are required)}
\Output{\begin{tabular}[t]{ll}	$C_1$, \ldots, $C_n$ &  (sets of cumulative clauses to be added) \\
			$V_1$, \ldots, $V_{n}$ & (sets of volatile clauses to be pushed or popped) \end{tabular}}

\BlankLine
$V_1 \longleftarrow F_1 \setminus F_2;\quad$
$C_1 \longleftarrow F_1 \setminus V_1$\;

\BlankLine
\For{$i \leftarrow 2$ \KwTo $n-1$}{
	$V_i \longleftarrow F_i \setminus F_{i+1}$\;
	$C_i \longleftarrow (F_i \setminus F_{i-1}) \setminus V_i$\;
	\ForEach{$c \in V_i \cap F_{i-1}$}{ 
		\For{$j \leftarrow 1\ \KwTo\ i-1$}{
			\If{$c \in C_j$}{
				$C_j \longleftarrow C_j \setminus \{c\}$\; 
				\For{$k = j\ \KwTo\ i-1$}{
					$V_k \longleftarrow V_k \cup \{c\}$\;
				}
				\Break\;
			}
		}
	}
}

\BlankLine
$C_n \longleftarrow F_n \setminus F_{n-1};\quad$
$V_n \longleftarrow \emptyset$\;
\BlankLine
\BlankLine

 \caption{Identifying cumulative and volatile clauses.}
\label{alg:cumvol}
\end{algorithm}


\paragraph{Analyzing PCNFs.}

For sequences of QBFs,
additionally the differences between quantifier prefixes must be identified.
Two quantified sets $Q$ and
$Q'$ are \emph{matching} iff $Q \cap Q' \neq \emptyset$.
Prefix $R$ is \emph{update-compatible to}  prefix $S$ iff all of the following conditions hold:
(i) for any quantified set of $R$, there is at most one matching quantified set in $S$;
(ii) if $P$ is a quantified set of $R$ and $Q$ is a matching quantified set in $S$, then $\quant(P) = \quant(Q)$; and
(iii) for any two  quantified sets $P_{1}$ and $P_{2}$ in $S$ with matching
quantified sets $Q_{1}$ and $Q_{2}$ in $R$, respectively, if the nesting level of $P_{1}$ is less than the nesting level $P_{2}$, then 
the nesting level of $Q_{1}$ is less than the nesting level of $Q_{2}$.

The instructions to  update quantifier prefixes are adding a
quantified set at a given nesting level or adding a variable to a
quantified set at a given nesting level. 
Update compatibility between prefixes $R$ and $S$ guarantees that
there is a sequence of instructions to turn $R$ into $S$ after unused
variables and empty quantified sets have been deleted by the QBF solver.
In particular, Condition~(i) guarantees that there is  no ambiguity when mapping quantified sets from the prefixes, (ii) expresses that quantifiers cannot change, and
(iii) states that quantified sets cannot be swapped.
The algorithm to generate update instructions first checks if
two quantifier prefixes $R$ and $S$ are update-compatible.
If this is the case,  then update instructions are 
computed as illustrated by Algorithm~\ref{alg:prefix}. 
\begin{algorithm}[t]
\small
\SetKwInOut{Input}{Input}\SetKwInOut{Output}{Output}
\SetKw{Move}{move}
\SetKw{from}{from}
\SetKw{Break}{break}
\SetKw{Write}{print}

\Input{Prefix $R$ and $S$  ($R$ has to be update-compatible to $S$)}
\Output{Instructions to update $R$ to $S$}
\BlankLine

$n \longleftarrow 0;\quad$ 
$m \longleftarrow 0$\;

\ForEach{quantified set $Q$ in $S$ from left to right}{

	\uIf {$Q$ has a matching quantified set $M$ in $R$} {
	
		$m\longleftarrow n\ + $\ nesting level of $M$ in $R$\;
		\Write{``Add literals $Q \setminus M$ to quantified set at nesting level $m$.''}\;
	
	}\Else{
	
		$n \longleftarrow n + 1$\;
		$m \longleftarrow m + 1$\;
		\Write{``Add quantified set $Q$ at nesting level $m$.''}\;
	
	}

}

 \caption{Generating update instructions for quantifier prefixes.}
\label{alg:prefix}
\end{algorithm}

\section{Case Studies}{\label{sec:benchmarks}}

In this section, we showcase 
our approach using different hardware verification problems for both SAT and QBF 
solvers. 
Benchmark problems consist of sequences of formulas that were either generated by a model-checking tool or that
were taken from existing benchmark collections where the original application is not available.

\paragraph{SAT: Bounded-Model Checking for Hardware Verification.}
We consider  benchmarks used for the single safety property track of the last Hardware Model Checking Competition (HWMCC 2014)\footnote{\url{http://fmv.jku.at/hwmcc14cav/}}. 
Based on the CNFs generated by the BMC-based model checker \aigbmc\footnote{Part of the AIGER package (\url{http://fmv.jku.at/aiger/)}}, we use our tools to generate incremental solver calls and compare different SAT solvers that implement the \ipasir interface. 
We used the SAT solvers 
\minisat~(v.220)~\cite{DBLP:conf/sat/EenS03},
\picosat~(v.961)~\cite{DBLP:journals/jsat/Biere08}, and
\lingeling~(v.ayv)~\cite{conf/sat/Biere14} as well as the QBF solver \depqbf~(v.4) for the considered problems.
All experiments  were performed on an AMD Opteron 6238 at 2.6~GHz under 64-bit Linux with a time limit of 3600 seconds and a memory limit of 7 GB.

\begin{table}[t]

\caption{Summary of different SAT solvers on hardware verification problems. }
\label{tab:results-sat}
\begin{center}\small

\begin{tabular}{|l|c|c|c|c|c|}
\hline
& \#problems &   {\minisat}      &       {\picosat} & {\lingeling} 
                                                                &   {\depqbf}      \\ 
\hline
BMC problems unrolled by 50 steps & 11 & 284 / 7 & 216 / 3 & 276 / 7 & 190 / 1 \\
BMC problems unrolled by 100 steps & 28 & 905 / 14 & 754 / 4 & 872 / 19 & 491 / 2  \\
\hline 
\end{tabular}

\end{center}
\vspace{-2\baselineskip}
\end{table}

Table~\ref{tab:results-sat} summarises the results.  
For each solver and problem class, numbers $m$ / $n$ mean that $m$ formulas in total were solved within the time limit, 
and $n$ is the number of problems where the maximal number of formulas among all other solvers could be solved.
For example, the first line summarises the results for  BMC problems that were unrolled by 50 steps. 
There are 11 problems in this class, thus 550 formulas in total. From these formulas, \minisat could solve 284 formulas, and for 7 out of 11 problems, no other
solver could solve more formulas than \minisat.
Not surprisingly, all SAT solvers  outperform the QBF solver \depqbf but there are few cases where \depqbf can compete. 
\minisat solves most formulas in total while \lingeling dominates on most benchmarks. More detailed experimental results can be found in the appendix.
The average time for our analyzing algorithm was 522 seconds. 
The number of clauses in the original sequences ranged from 
$2.3$ 
to 
$56.3$ million 
with an average of around $19$ million clauses.
The inputs for the benchmarking program that represent only the update instructions comprise only 1.2 million clauses on average which shows
that we obtain a quite compact representation of incremental benchmarks.
We have submitted all problems from Table~\ref{tab:results-sat} to the Incremental Library Track of the SAT Race 2015.

\paragraph{QSAT: Partial Design Problems.}
\begin{wraptable}{I}{0.60\textwidth}
\vspace{-0.75cm}
\caption{
QBF solvers on incomplete  design problems.
}
\label{tab:results-qbf}
\begin{center}
\small
\begin{tabular}{|l|c|rr|rrr|}
\hline
\multirow{2}{*}{Benchmark} & \multirow{2}{*}{$k$} & \multicolumn{2}{|c|}{\scriptsize{non-incremental}} & \multicolumn{3}{|c|}{\scriptsize{incremental}} \\
					&				& {\scriptsize\qube}	&	{\scriptsize\depqbf} &	 {\scriptsize\qube (fwd)}& {\scriptsize\qube (bwd)}	& {\scriptsize\depqbf} 		 \\
\hline\hline
enc04	&	17	&	3		&	3	&	3		&	2		&	{1}	\\	 				
enc09	&	17	&	7		&	5	&	7		&	4		&	{3}	\\	 				
enc01	&	33	&	31		&17   	&	28		&	24		&	{5}	\\	 
enc03	&	33	&	33		& {16}	&	289		&	28		&	27	\\	 				
enc05	&	33	&	64		&24		&	61		&	46		&	{7}	\\	 				
enc06	&	33	&	29		&26		&	28		&	24		&	{10}	\\	 				
enc07	&	33	&	75		&16		&	76		&	69		&	{5}	\\	 				
enc08	&	33	&	108	        &16		&	110		&	79		&	{5}	\\	 				
enc02	&	65	& 	271	        &{106}	&	TO		&	269		&	175\\	 				
\hline
tlc01		&	132	&	26		&	68		&	133		&	130		&	{17}	\\	 				
tlc03		&	132	&	24		&	160		&	8		&	{8}		&	17	\\	 				
tlc04		&	132	&	769		&	2196		&	1204		&	27		&	{25}	\\	 					
tlc05		&	152	&	1330		&	4201		&	2057		&	38		&	{34}	\\	 				
tlc02		&	258	&	MO		&	TO		&	MO		&	{98}		&	1908 \\	 					
\hline
\end{tabular}
\end{center}
\vspace{-2.1\baselineskip}
\end{wraptable}
To illustrate our approach in the context of QBF solving, we consider the problem of \emph{verifying partial designs}, i.e.,
sequential circuits  where parts of the specification are black-boxed.
In recent work~\cite{DBLP:conf/date/MarinMLB12,DBLP:journals/aicom/MillerMB15},  the question
whether a given safety property can be violated  regardless of the implementation of  a black-box has been translated to QBFs which are  solved incrementally 
by a version of the  QBF solver  \qube~\cite{DBLP:journals/jsat/GiunchigliaMN10}.
Benchmarks 
are available from QBFLIB,\footnote{\url{http://www.qbflib.org}} however
neither the solver used
in~\cite{DBLP:conf/date/MarinMLB12,DBLP:journals/aicom/MillerMB15} nor the
application program used to generate sequences of QBFs are publicly available. 
Marin et al.~\cite{DBLP:conf/date/MarinMLB12} introduced two encoding strategies: forward incremental  and backward incremental reasoning. In a nutshell, the quantifier prefix is always extended to the right in the  former approach, while it is extended to the left in the latter approach. 
Both strategies yield the same sequences of formulas up to renaming~\cite{DBLP:conf/date/MarinMLB12}.
We used the publicly available instances from the forward-incremental encoding without preprocessing to evaluate \depqbf.
Instances from the backward-incremental approach are not publicly available.

Table~\ref{tab:results-qbf} shows the comparison between \qube and \depqbf. 
Runtimes are in seconds, $k$ is the index of the first satisfiable formula, TO and MO refer to a timeout and memout, respectively. 
The maximal runtime of Algorithm~\ref{alg:cumvol} and~\ref{alg:prefix} was 95 seconds.
Runtimes for \qube in Table~\ref{tab:results-qbf} are the ones  reported in~\cite{DBLP:conf/date/MarinMLB12}. 
There, experiments were carried out on an AMD Opteron 252 processor running at 2.6 GHz
with 4GB of main memory and a timeout of 7200 seconds.
Experiments for \depqbf 
were performed on a 2.53 GHz Intel Core 2 Duo processor with 4GB of main memory with
OS X 10.9.5 installed. Thus runtimes are not directly comparable because experiments were carried out on different machines, they
give, however, a rough picture of how
the solvers relate.
Like \qube, \depqbf benefits from the incremental strategy on most instances.  
The backward-incremental strategy is clearly the dominating strategy for \qube.
A quite eye-catching observation is that forward-incremental solving, while
 hardly improving the performance 
of \qube compared to the non-incremental approach, works quite well for \depqbf. 

\section{Conclusion}{\label{sec:concl}}

We presented an approach to automated benchmarking of incremental SAT and QBF
solvers by translating sequences of formulas into API calls of incremental SAT and QBF solvers
executed by a benchmarking program. Several incremental solvers may be tightly
integrated into the benchmarking program by linking them as libraries. 
Thus,  we decouple the generation of formulas by an
application from the solving process which is particularly relevant when application
programs  are not available.
Additionally, we make sequences of formulas which already
exist in public  benchmark collections available for benchmarking
and testing. 
We illustrated our approach to automated benchmarking of incremental SAT and
QBF solvers on instances from hardware verification problems. 
To improve the
performance of incremental QBF solving on these problems, we want to integrate
incremental preprocessing {into} \depqbf. As shown
in~\cite{DBLP:conf/date/MarinMLB12,DBLP:journals/aicom/MillerMB15},
preprocessing potentially improves the performance of incremental workflows
considerably.



\newpage

\begin{appendix}

\section{Correctness of Algorithms~1 and 2}
\begin{theorem}
Algorithm~\ref{alg:cumvol} is totally correct with respect to the precondition that $\sigma = (F_{1}, \ldots, F_{n})$ is a sequence of sets of clauses with $n \geq 2$ and the
postcondition that any $C_{i}$, $1 \leq i \leq n$, contains the cumulative clauses of $F_{i}$ in $\sigma$, and 
any $V_{i}$,  $1 \leq i \leq  n$, contains the volatile clauses of $F_{i}$ in $\sigma$.
\end{theorem}

\begin{proof}
Clearly, Algorithm~\ref{alg:cumvol} terminates on each input.

We show that the condition that any $C_{j}$, $1 \leq j < i$, contains the cumulative clauses of $F_{j}$ in the subsequence $\sigma_{i} = (F_{1}, \ldots, F_{i})$ of $\sigma$, and 
any $V_{j}$,  $1 \leq j < i$, contains the volatile clauses of $F_{j}$ in $\sigma_{i}$ is an invariant of the main loop (at Line 2).
The invariant together with $i = n$ implies the postcondition as $C_{n}$ always contains those clauses that are in $F_{n}$ but not in $F_{n-1}$, and $V_{n}$ always equals the empty set (Line 12).
Likewise, the precondition implies the invariant since after Line 1, $V_ {1}$ contains all clauses of $F_{1}$ that are not in $F_{2}$ and which are thus volatile in $F_{1}$
in the sequence $F_{1}, F_{2}$, and $C_{1}$ contains all clauses which are in $F_{1}$ and $F_{2}$ and which are hence cumulative in $F_{1}$ in the sequence $F_{1}, F_{2}$. 

It remains to show that if the invariant holds for some $i$, $2 \leq i \leq n-1$ at Line 3, then it holds for $i+1$ after executing Lines 3--11.  
After Line 3, $V_{i}$ contains all the clauses that are volatile in $F_{i}$ in $\sigma_{i+1}$.  
Likewise, after Line 4, $C_{i}$ contains all the clauses that are in $F_{i}$ but not in $F_{i-1}$ and which are not volatile in $F_{i}$, that is, which are cumulative in $F_{i}$ in $\sigma_{i+1}$. 
Note that if a clause $c$ is volatile in some $F_{j}$, $j < i$, in $\sigma_{i}$, then $c$ is also volatile in $F_{j}$ in $\sigma_{i+1}$. 
On the other hand, if a clause is cumulative in $F_{j}$, it can be the case that $c$ becomes volatile in $\sigma_{i+1}$ if $c \not\in F_{i+1}$.
Hence, it is possible that clauses that were previously classified as cumulative need to be reclassified. 

We make use of the following claim: After Line 4, a clause $c$ is in 
$V_{i} \cap F_{i-1}$ iff, for some $j < i$, $C_{j}$ contains a clause $c$ that is volatile in $\sigma_{i+1}$. 
This claim is proven as follows:
Assume that for some $j < i$, $C_{j}$ contains a clause $c$ that is volatile in $\sigma_{i+1}$. As the invariant holds for $i$, $c \in F_{k}$, for all
$j \leq k \leq i$ but  $c \not\in F_{i+1}$ and thus $c \in V_{i}$. Clearly, $c \in V_{i} \cap F_{i-1}$.
On the other hand, assume some clause $c$ is in $V_{i} \cap F_{i-1}$. Clearly, $c \in V_{i}$ implies $c \in F_{i}$. Hence, as the invariant holds for $i$, $c \in C_{j}$, for some $j < i$, and, since $c \in V_{i}$, c is volatile in $F_{j}$ in $\sigma_{i+1}$.

By virtue of the above claim, $V_{i} \cap F_{i-1}$ contains precisely those clauses which need to be reclassified as volatile.
After Lines  6 -- 11, for each $c \in V_{i} \cap F_{i-1}$, the first (and only) $C_{j}$ with $c \in C_{j}$ is found, and $c$ is removed from $C_{j}$ and added to 
all $V_{l}$, $j \leq l \leq i-1$. 
Hence, after Lines  3 -- 11, the invariant holds for $i+1$.\qed
\end{proof}

Algorithm~\ref{alg:prefix} works as follows. 
For each quantified set $q$ of $S$, either it has one matching set $q'$ in $R$ (Lines 4 and 5) or it does not have a matching quantified set in $R$ (Lines 7,8, and 9).
In the former case, we need to add the atoms in $q$ to $q'$ if they are not already there (Line 5). In the latter case, we need to add the entire set $q$ to the prefix (Line 9). 
Adding atoms and quantified sets is always done at the right nesting level $m$. We store in $n$ the number of new quantified sets that have been added. At Line 5, when adding atoms to a matching quantified set, $m$ is
the nesting level of the matching quantified set in $R$ plus the number $n$ of previously added unmatched quantified sets. At Line 9, when adding an entire quantified set, $m$ is the nesting level of the quantified set that was modified last plus one. 

\begin{table}[ht!]

\caption{Detailed results for Table~\ref{tab:results-sat}: SAT solvers on hardware verification problems. }
\label{tab:results-sat-detailed}
\begin{center}\small

\begin{tabular}{cc}

{ \small
\begin{tabular}{|l|c|c|c|c|}
\hline
 &   {\scriptsize\minisat}      &       {\scriptsize\picosat} & {\scriptsize\lingeling} 
                                                                &   {\scriptsize\depqbf}      \\ 
\hline
6s393r  & 15 & 15  & 14 &  10 \\ 
6s394r  & 27 & 23  & 22 &  16 \\ 
6s514r  & 17  & 16  & 16 &  14 \\ 
arbi0s08  & 12 &  12  & 13 &  11 \\ 
arbi0s16  & 17  & 17  & 17 &  17 \\ 
arbixs08  & 9  & 10  & 10 &  10 \\ 
cuabq2f  & 28  & 22  & 40 &  19 \\ 
cuabq2mf  & 77  & 46  & 65  & 20 \\ 
cuabq4f  & 21  & 22  & 26 &  15 \\ 
cuabq4mf  & 25  & 23  & 40 &  16 \\ 
cuabq8f  & 22  & 21  & 26 &  19 \\ 
cubak  & 83  & 43  & 55 &  28 \\ 
cufq2  & 101  & 88  & 83 &  21 \\ 
cugbak  & 38  & 32  & 42 &  21 \\ 
cuhanoi10  & 35  & 34  & 35  & 17 \\ 
cujc12  & 47  & 46  & 38 &  15 \\ 
cunim1  & 22  & 21  & 23 &  19 \\ 
cunim2  & 22  & 21  & 22 &  18 \\ 
cuom1  & 14  & 13  & 14 & 10 \\ 
cuom2  & 14  & 13  & 15 &  10 \\ 
\hline 
\multicolumn{5}{c}{BMC problems unrolled by 100 steps.} \\
\end{tabular}
}

&

{
\small
\begin{tabular}{|l|c|c|c|c|}
\hline
    &   {\scriptsize\minisat}      &       {\scriptsize\picosat} & {\scriptsize\lingeling} 
                                                                &   {\scriptsize\depqbf}      \\ 
\hline
cuom3  & 15  & 14  & 16 &  10 \\ 
cupts14  & 34  & 32  & 37 &  25 \\ 
cupts15  & 35  & 32  & 37 &  27 \\ 
cupts16  & 35  & 34  & 37 &  24 \\ 
cutarb16  & 52  & 37  & 48 &  30 \\ 
cutf1  & 44  & 32  & 35 &  21 \\ 
pdtfifo1to0  & 16  & 16  & 16 &  13 \\ 
pdtpmsdc16  & 28  & 19  & 30 &  15 \\ 
\hline
\multicolumn{5}{c}{BMC problems unrolled by 100 steps.} \\

\hline
6s188  & 39  & 34  & 36 &  33 \\ 
6s24  & 25  & 23  & 27 &  20 \\ 
6s270b1  & 51  & 13  & 51 &  11 \\ 
arbi0s32p03  & 32  & 32  & 33 &  32 \\ 
arbixs16p03  & 16  & 16  & 16 &  16 \\ 
bmhan1f1  & 29  & 21  & 29 &  19 \\ 
bobpcihm  & 17  & 15  & 16 &  11 \\ 
bobsmvhd3  & 14  & 11  & 14 &  10 \\ 
cufq1  & 45  & 33  & 37 &  23 \\ 
cujc128  & 7  & 8  & 7 &  7 \\ 
cujc32  & 9  & 10  & 10 &  8 \\ 
\hline
\multicolumn{5}{c}{BMC problems unrolled by 50 steps.} \\

\end{tabular}
}

\end{tabular}

\end{center}
\vspace{-\baselineskip}
\end{table}

\end{appendix}

\end{document}